\documentclass[reprint, review]{elsarticle}
\usepackage{lineno,hyperref}
\modulolinenumbers[1]

\usepackage{relsize}
\usepackage{amsmath,amsthm}
\usepackage{amssymb,latexsym}
\usepackage{graphicx}
\usepackage[labelsep=endash]{caption}
\usepackage{mathptmx}
\usepackage[utf8]{inputenc}
\usepackage{subcaption}
\newtheorem{thm}{Theorem}[section]

\theoremstyle{definition}

\begin{document}
\makeatletter
\def\ps@pprintTitle{%
   \let\@oddhead\@empty
   \let\@evenhead\@empty
   \def\@oddfoot{\reset@font\hfil\thepage\hfil}
   \let\@evenfoot\@oddfoot
}
\makeatother
\begin{frontmatter}

\title{ Scrutiny of stagnation region flow in a nanofluid suspended permeable medium due to inconsistent heat source/sink}
\author[1]{Rakesh Kumar}
\ead[a]{rakesh@cuhimachal.ac.in}
\author[1]{Ravinder Kumar}
\address[1]{Department of Mathematics, Central University of Himachal Pradesh, Dharamshala, India}
\begin{abstract}
In present analysis, nanofluid transport near to a stagnation region over a bidirectionally deforming surface is scrutinized. The region is embedded with Darcy-Forchheimer medium which supports permeability. The porous matrix is suspended with nanofluid, and surface is under the influence of inconsistent heat source/sink.  Using similarity functions, framed governing equations are switched to a collection of ordinary differential equations. Output is procured via optimal homotopy asymptotic method (OHAM). Basic notion of OHAM for a vector differential set-up is presented along with required convergence theorems. At different flow stagnation strengths, nanofluid behavior is investigated with respect to variations in porosity parameter, Forchheimer number, Brownian motion, stretching ratio, thermophoretic force, heat source/sink and Schimdt number. Stagnation flow strength invert the pattern of boundary layer profiles of primary velocity. Heat transfer has straightforward relation with Forchheimer number when stagnation forces dominate stretching forces. 
\end{abstract}
\end{frontmatter}

\section{Introduction}
Three dimensional stagnation flow (having three velocity components) inevitably appears when mass of fluid falls on a solid body (flat / non-flat) or  when solid roaming surfaces in a viscous fluid are considered. This type of flow is possible for bodies of all shapes because in this case the flow near stagnation region can be simulated by the tangent plane \cite{borrelli2013}. This flow has marvelous applications in numerous  hydrodynamics processes like cooling of electronic apparatus, drag devaluation and radial diffusers.
Heat and mass transport analysis of such flows over deforming surfaces when suspended with nanoparticles is a progressive area of research in fluid dynamics due to the associated applications in bio science, nuclear science and industries \cite{ramesh}. Some works on stagnation flow (three dimensional) of nanofluid over various geometries are given in \cite{bachok2010}-\cite{nadeemabbas2018}.\\

In thermodynamics, any object which is used to produce heating/cooling is considered as a source / sink (natural/man-made) such as Sun, air, electricity, air conditioner and heater. Heat transport dynamics due to heat generation/absorption has multiple applications in data center cooling, laboratory cooling, health care air conditioning, food store, heating water or other liquids, increasing the temperature of metal pieces, environmental conditioning, food processes. Heat source/sink dynamics is crucial in controlling the temperature of industrial devices as indicated by Singh and Kumar \cite{singhkumar2009}. Ram et al. \cite{pram2017} studied heat generation effects on the flow of MHD fluid along heated plate (fluctuating cosinusoidally). Hayat et al. \cite{hayatandshehzad} examined MHD Maxwell flow accounting heat source/sink due to the variable nature of thermal conductivity, and exploited homotopy analysis method. Hayat et al. \cite{hayatandalsaidi} also inspected Jeffrey fluid transport along a surface stretching in two lateral directions considering heat source/sink. In practical approach, heat generation/absorption can be considered as a function of space and time variables in order to optimize heat flow field \cite{sulochna}. Chamkha et al. \cite{chamkhapof2017} examined the effects of heat source/sink locations inside a nano-powder filled porous enclosure and found that convective heat transfer is curtailed with the addition of nanoparticles. Some models on variable heat source/sink are developed in \cite{raju}-\cite{eldahab2004}.    \\

Other pertinent aspect in fluid dynamics is flow due to nanofluid suspended porous medium where flow and heat distribution is substantially influenced by its presence. A porous medium is said to be permeable if the pores are inter-connecting, and is characterized by its porosity (an empty space in the medium through which fluid is allowed to flow). High porosity of a medium is a natural requirement to deal with wider cross sectional areas of the medium and related higher velocities of the fluid. Here Darcy-Forchheimer model is an improved model which assists the flow over a surface considering a square velocity factor in momentum equations and high porosity. Shehzad et al. \cite{shehzadandabbasi} exploited Darcy-Forchheimer medium  to investigate an Oldroyd-B liquid flow. Hayat et al. \cite{hayatdarcy} discussed flow of carbon nanotubes (three dimensional) through Darcy-Forchheimer medium.  \\

Three dimensional boundary layer equations used to simulate nano-powder transport are highly nonlinear in nature due to presence of inertial terms in these equations (\cite{jcis2017}-\cite{joml2019}). Further complexity is added by the inclusion of nanoparticles in traditional fluids and considering Brownian motion and thermophoretic force effects (\cite{rip2018}, \cite{drrakesh1}). For such equations, there exists is no standard procedure to obtain exact solutions. Thus researchers are trying to develop analytical schemes which provide accurate solutions. One strong method among the class of known methods is OHAM. This method provides approximate analytical solution in a series format quite easily. Main feature of this technique is the control over the region and parameters of convergence \cite{vasile2009}. Due to this, OHAM is revisited for present investigation. \\

As per our knowledge there is no study on the scrutiny of a three dimensional stagnation flow of nanofluid over a stretchable sheet  along with the consideration of  Darcy-Forchheimer model for porous matrix and inconsistent heat source/sink. Here, basic theorems are also developed by revisiting OHAM for accuracy and convergence, and presenting a treatment for vector differential equations.

\section{Geometrical and mathematical formulations}
Let us take into account the flow which creates a stagnation region on one side of a finite flat sheet and is submerged in nanofluid suspended medium as shown in figure \ref{geometry}. Let $D^*=(0,l]\times (0,m] \times [0, \infty)$ be the domain of flow field and $\partial D^*$ be its boundary. Let surface stretching velocities be $U_w(x,y)=U_0(x+y+a)$ and $V_w(x,y)=U_1(x+y+a)$ in $x$ and $y$ directions respectively for constants $U_0$ and $U_1$. Let straining velocities at free stream be $U_\infty(x,y)=V_\infty(x,y)=U_2(x+y+a)$ for constant $U_2$. We choose a fixed coordinate system so that the origin may be treated as a stagnation point.  \\
\begin{figure}[!t]
\begin{center}
{\includegraphics[scale=.25]{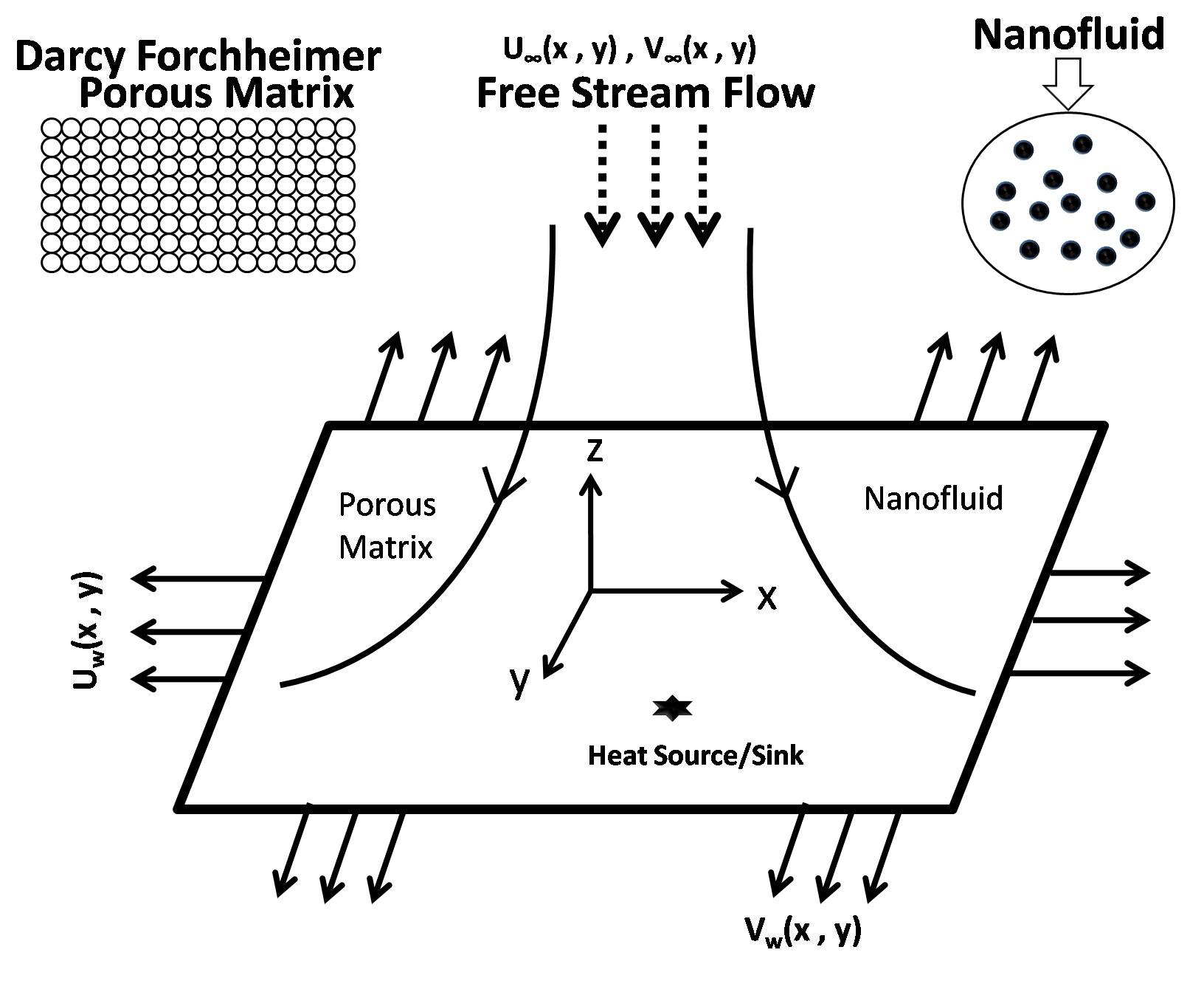}}
\end{center}
\caption{Geometrical representation of the problem }\label{geometry}
\end{figure}
Let $I_c$ be the inertia coefficient due to Darcy-Forchheimer porous medium which is defined as:
\begin{align}
I_c=\frac{c_b}{(x+y+a)\sqrt{{K_1}}},
\end{align}
where $c_b$ is the drag coefficient.\\
The inconsistent heat source/sink (having space/temperature dependency) in a region near to stagnation point embedded in Darcy Forchheimer medium is approximated as (\cite{mabood}, \cite{eldahab2004}):
\begin{align}
q'''=\dfrac{kU_w}{ (x+y+a)\nu_f}\left[Q_1^*(T_w-T_\infty)\exp \left({-z \sqrt{\left(\dfrac{U_0}{\nu_f}\right)}}\right)+Q_2^*(T-T_\infty))\right],
\end{align}
where $Q_1^*$, $Q_2^*$ are constant coefficients of heat generation/absorption.\\
The system of governing equations refined under boundary layer approximations and pressure elimination is produced as (\cite{hayattransformation}-\cite{mustafa}):  
\begin{align}\label{bl1}
u_x+v_y+w_z&=0 \quad \text{in} \,\,\, D^*,
\end{align}
\begin{align}\label{bl2}
u u_x+v u_y+w u_z=&U_\infty({U_\infty})_x+V_\infty({U_\infty})_y+\nu_{f}\underset{z}{\Delta}u -(u-U_\infty)\left[({\nu_{f}}/{K_1})+I_c(u+U_\infty)\right] \quad \text{in} \,\,\, D^*,
\end{align}
\begin{align}\label{bl3}
u v_x+v v_y+w v_z=&U_\infty({V_\infty})_x+V_\infty({V_\infty})_y+\nu_{f}\underset{z}{\Delta}v -(v-V_\infty)\left[({\nu_{f}}/{K_1})+I_c(v+V_\infty)\right] \quad \text{in} \,\,\, D^*,
\end{align}
\begin{align}\label{bl4}\nonumber
uT_x+vT_y+wT_z&=\left[{k_{f}}/{(\rho c_p)_{f}}\right]\underset{z}{\Delta}T+\left[{(\rho c_p)_{np}}/{(\rho c_p)_f}\right]\left(D_B C_z T_z+{D_T}\left(T_z\right)^2/{T_{\infty}} \right)&\\
&+\frac{q'''}{(\rho c_p)_f} \quad \text{in} \,\,\, D^*,&
\end{align}
\begin{align}\label{bl5}
u C_x+v C_y+w C_ z=D_{B}\underset{z}{\Delta}C+({D_T}/{T_{\infty}})\underset{z}{\Delta}T \quad \text{in} \,\,\, D^*
\end{align} 
which satisfies following conditions on $\partial{D^*}$ (boundary of $D^*)$:
\begin{equation}
\left.
\begin{aligned}
(u, v, w, T, C)&=(U_w(x,y), V_w(x,y), 0, T_w, C_w) \quad &\mbox{for} \quad &z=0\\
(u, v, T, C)&=(U_{\infty}(x,y), V_\infty(x, y), T_\infty, C_\infty) \quad &\mbox{for} \quad &z\rightarrow \infty
\end{aligned}
 \right\}.
\end{equation}
Here $\underset{z}{\Delta}$ is a one variable Lapacian operator. 
Let us choose following similarity transformations \cite{drrakesh2}:
\begin{equation}
\left.
\begin{aligned}
&\eta={z}\left(\dfrac{U_0}{\nu_f}\right)^{0.5}, \quad \psi=\left({\nu U_0}\right)^{0.5}{F_1(\eta)}, \quad u=U_0(x+y+a) F_{1}',\quad v=U_0(x+y+a) F_{2}', &\\
& T=T_{\infty}+(T_w-T_{\infty})\Theta, \quad C=C_{\infty}+(C_w-C_{\infty})\Phi,\quad w=-\left({\nu U_0}\right)^{0.5}\left(F_1+F_2)\right)&\\
\end{aligned}
 \right\}
\end{equation}
With the aid of these variations, nanofluid flow satisfies following ordinary differential equations:
\begin{align}\label{e1}
F_1'''+(F_1+F_2)F_1''-F_1'(F_1'+F_2')-KF_1'-Fr F_1'^2+V_r\left\{(2n+Fr)V_r+K\right\}=0,
\end{align} 
 \begin{align}\label{e2}
F_2'''+(F_1+F_2)F_2''- F_2'(F_1'+F_2')-K F_2'-Fr F_2'^2+V_r\left\{(2n+Fr)V_r+K\right\}=0,
\end{align} 
\begin{align}\label{e3}
\Theta''+Pr(F_1+F_2)\Theta'+Pr(Nb\Theta'\Phi'+Nt\Theta'^2)+(Q_1^* e^{-\eta}+Q_2^* \Theta)=0,
\end{align} 
 \begin{align}\label{e4}
 \Phi''+Sc(F_1+F_2)\Phi'+\frac{Nt}{N_b}\Theta''=0
 \end{align}
 and boundary conditions 
\begin{equation}\label{b1}
\left.
\begin{aligned}
 (F_1+F_2, F_1', F_2', \Theta, \Phi)&=\left(0, 1, \lambda, 1, 1 \right) \quad &\text{at} \quad &\eta=0, \\
(F'_1, F_2', \Theta, \Phi)&=(V_r, V_r, 0, 0) \quad &\text{as} \quad &\eta \rightarrow\infty
\end{aligned}
 \right\}.
\end{equation}

 The expressions for parameters appeared in the above equations are:
\begin{equation}\nonumber
\left.
 \begin{aligned}
&Fr\text{(Forchheimer parameter)}=\dfrac{c_b}{\sqrt{K_1}},\hspace{0.4cm} K\text{(porous medium permeability)}=\dfrac{\nu_f}{U_0 K_1},   \\& V_r\text{(stretching velocity ratio parameter)}= \dfrac{U_2}{U_0},\hspace{0.4cm}Nb\text{(Brownian motion number)}=\dfrac{\tau D_B(C_w-C_\infty)}{\nu_f},\\& Sc\text{(Schmidt number)}=\dfrac{\nu_f}{D_B},\hspace{0.4cm}  Nt\text{(thermophoresis number)}=\dfrac{\tau D_T(T_w-T_\infty)}{T_\infty \nu_f}, \hspace{0.4cm} Pr\text{(Prandtl number)}=\dfrac{\nu_f}{\alpha_f}, \\& \lambda\text{(stretching velocity ratio parameter)}= \dfrac{U_1}{U_0}.\hspace{1.6cm} 
 \end{aligned}
 \right.
\end{equation}

In fluid dynamics, practical aspects of the problem are explored through coefficients of skin friction $(C_{sfx}, C_{sfy})$, Nusselt number$(Nu_x)$ and Sherwood number$(Sh_x)$ are  taken as in the work of Kumar et al. \cite{drrakesh3}
\begin{align*}
C_{sf_{x}}=\dfrac{\mu_f}{\rho_f U_w^2}\left(u_z\right)_{z=0},\hspace{2cm} C_{sf_{y}}=\dfrac{\mu_f}{\rho_f V_w^2}\left(v_z\right)_{z=0},\hspace{1cm}
\end{align*}
\begin{align*}
Nu_x=\dfrac{-(x+y+a)}{(T_w-T_\infty)}\left(T_z\right)_{z=0},\hspace{0.9cm} Sh_x=\dfrac{-(x+y+a)}{(C_w-C_\infty)}\left(C_z\right)_{z=0}.
\end{align*}
The dimensionless form of these expressions will be
\begin{align*}
\sqrt{Re_\chi}C_{F_{1}{\chi}}=F_1''(0),\hspace{1.6cm} \sqrt{Re_\chi}C_{F_{2}{\chi}}=F_2''(0),\hspace{0.6cm}
\end{align*}
\begin{align*}
\dfrac{Nu_{\chi}}{\sqrt{Re_\chi}}=-\Theta' (0),\hspace{1.8cm}\dfrac{Sh_{\chi}}{\sqrt{Re_\chi}}=-\Phi' (0), \hspace{1cm}
\end{align*}
where $Re_\chi \text{(local Reynolds number)}=\dfrac{U_w(x,y)\chi}{\nu_f}$ for $\chi=(x+y+a)$.

\section{Essential features of OHAM, convergence and application}
The coupled set of equations  (\ref{e1})-(\ref{e4}) with restrictions (\ref{b1}) is  solved using optimal homotopy asymptotic method. Previous studies indicate that authentic solutions can be produced if first or second order of approximations in this method is considered (\cite{maboodoham}-\cite{vajraveluoham}). Other nice aspect of this method is that the solution travels from initial guess to final approximate solution as homotopy parameter ($q$) is allowed to vary from $q=0$ to $q=1$. Next we present the basic approach of OHAM for vector differential equations, and to validate its accuracy, convergence theorems are also developed. \\

Now to explain the practice of OHAM, we look at differential equations (coupled system) of the type
\begin{align}\label{oe}
{\bf{A}}{\bf {F}}+{\bf{B}}{\bf{F}}+{\bf {G}}={\bf{O}} \quad \text{on} \quad D^*
\end{align}
with boundary restraints
\begin{align}\label{oe1}
{\bf {I}}\left( {\bf {F}}, {\bf {F'}}  \right)={\pmb{\gamma}} \quad \text{on}\quad \partial D^*.
\end{align}
For $j=1(1)m$, $H_j$ homotopies  
$H_j (\eta,q):D^*\times [0,1]\rightarrow \mathbb{R}$  satisfy following homotopy equation:
\begin{align}\label{Oham2}
&(1-q)[{\bf{A}}({\bf{F}})+{\bf {G}}_0]-H^*(q)[\left( {\bf{B}}({\bf{F}})+{\bf{A}}({\bf{F}}\right)+{\bf {G}}]={\bf{O}}, 
\end{align}
where $H^*(q)=N_1q+N_2q^2+\dots$ with $N_1, N_2,\dots$ as convergence parameters and $q\in[0,1],$ whenever $\eta\in D^*$.\\
In above equations, $D^*$, $\partial D^*$, ${\bf{A}}$, ${\bf{B}}$, ${\bf {I}}$ signify  modified domain, boundary of domain, a linear operator,  a non linear operator and boundary operator respectively, whereas  ${\bf {G}}_0$, ${\bf {F}}$, ${\bf {G}}$ present respectively initial guess, an unknown function and  known analytic function. The mentioned symbols have following matrix presentations: ${\bf{A}}={[A_{jj}]}_{m\times m}$, ${\bf{B}}={[B_{ij}]}_{m\times m}$, ${\bf {F}}={[f_{j}]}_{m\times 1}$, ${\bf {G}}={[g_{j}]}_{m\times 1}$, ${\bf {I}}={[I_{j}]}_{m\times 1}$, ${\pmb{\gamma}}={[\gamma_{j}]}_{m\times 1}$, ${\bf {G}}_0={[{g^0_j}]}_{m\times 1}$ and ${\bf {O}}={[0]}_{m\times 1}$ for $i,j=1(1)m$.\\
Here, homotopy equation (\ref{Oham2}) is returned to original form (equation (\ref{oe})) when $q=1$, and gives initial guess when when $q=0.$\\
Following OHAM, we take $s^{\text{th}}$-order approximation for ${\bf {F}}$ as
\begin{align}\label{ass}
{\bf {\overline{F}}}&={\bf{F}}_0(\eta)+q{\bf{F}}_1(\eta,N_1)+q^2{\bf{F}}_2(\eta,N_1,N_2)+\dots+q^s{\bf{F}}_s(\eta,N_1,N_2,\dots,N_s),
\end{align}
where ${\bf{\overline{F}}}={[\overline{f_{j}}]}_{m\times 1}$, ${\bf{F}}_0={[f_{0j}]}_{m\times 1}$, ${\bf{F}}_1={[f_{1j}]}_{m\times 1}$, \dots , ${\bf{F}}_s={[f_{sj}]}_{m\times 1}.$

Utilizing assumption (\ref{ass}) in (\ref{oe}) and (\ref{oe1}), the comparison of the coefficients of common powers of $q$ help to obtain the final approximate solution which involve convergence parameters as:
\begin{align}\label{seriessol}
{\bf {\overline{F}}}&=\lim_{q\to 1}\left({\bf{F}}_0(\eta)+q{\bf{F}}_1(\eta,N_1)+q^2{\bf{F}}_2(\eta,N_1,N_2)+\dots+q^s{\bf{F}}_s(\eta,N_1,N_2,\dots,N_s)+\dots\right).
\end{align}
For the determination of optimum values for convergence parameters, residual function is taken as 
\begin{align}
\mathbf{R}(\eta,N_1,N_2,\dots,N_s)={\bf{A}}{\bf {\overline{F}}}+{\bf{B}}{\bf{\overline{F}}}+{\bf {G}}
\end{align}
and corresponding functional as
\begin{align} 
{\mathbb{J}}(N_1,N_2,\dots,N_s)=\int^\infty_0 \sum^{m}_{j=1}R_j^2(\eta,N_1,N_2,\dots,N_s)d\eta,
\end{align}
where $\mathbf{R}$=$[R_1,R_2,\dots,R_m]^T$.\\
Using least square guidlines, convergence parameters $N_1,N_2,\dots,N_s$ can be  resolved through following set of equations:
\begin{align}
\dfrac{\partial {\mathbb{J}}}{\partial N_1}=\dfrac{\partial {\mathbb{J}}}{\partial N_2}=\dots=\dfrac{\partial {\mathbb{J}}}{\partial N_s}=0.
\end{align}

\subsection{Convergence of results through OHAM}
\begin{thm} 
Let ${\bf{F^*}}=\left\{{\bf{H}}_0,{\bf{H}}_1,{\bf{H}}_2, \dots,{\bf{H}}_s, \dots \right\}$ be the sequence of vector valued functions in Hilbert space where ${\bf{H}}_0={\bf{F}}_0(\eta)$, ${\bf{H}}_1={\bf{F}}_0(\eta)+{\bf{F}}_1(\eta,N_1)$, $\dots$, ${\bf{H}}_s={\bf{F}}_0(\eta)+{\bf{F}}_1(\eta,N_1)+.....+{\bf{F}}_s(\eta,N_1,N_2,...N_s).$ Then the series solution (\ref{seriessol}) converges if $\exists$ a constant $\delta$ $(0<\delta<1)$ such that 
\begin{align}
{\bf{F}}_{s+1}(\eta,N_1,N_2,\dots,N_{s+1})\leq \delta {\bf{F}}_s(\eta,N_1,N_2,\dots,N_s) \quad \forall \quad s\geq m\quad \text{for some}\quad s,m\in \mathbb{N} 
\end{align}
\end{thm}
\begin{proof}
Consider 
\begin{align*}
\rVert{{\bf{H}}_{s+1}}-{\bf{H}}_s\rVert&=\rVert {\bf{F}}_{s+1}(\eta,N_1,N_2,\dots,N_{s+1}), \rVert\\&
\leq \delta\rVert {\bf{F}}_s(\eta,N_1,N_2,\dots,N_s) \rVert\\& \leq \delta^2\rVert {\bf{F}}_{s-1}(\eta,N_1,N_2,\dots,N_{s-1}) \rVert\\
&\vdots\\&
\leq \delta^{s-m+1}\rVert {\bf{F}}_m(\eta,N_1,N_2,\dots,N_m) \rVert 
\end{align*}
Now $\forall$ \hspace{0.2cm} $s,m\in \mathbb{N}$ and $s>s^*>m,$
\begin{align}\nonumber
\rVert {\bf{H}}_{s}-{\bf{H}}_{s^*}  \rVert&=\rVert ({\bf{H}}_s(\eta)-{\bf{H}}_{s-1}(\eta))+({\bf{F}}_{s-1}(\eta)-{\bf{H}}_{s-2}(\eta))\\
&+({\bf{H}}_{s-2}(\eta)-{\bf{H}}_{s-3}(\eta))+\dots+({\bf{H}}_{s^*+1}(\eta)-{\bf{H}}_{s^*}(\eta)) \rVert\\\nonumber
 &\leq \rVert ({\bf{H}}_s(\eta)-{\bf{H}}_{s-1}(\eta))\rVert+\rVert({\bf{H}}_{s-1}(\eta)-{\bf{H}}_{s-2}(\eta))\rVert\\&\nonumber+\rVert({\bf{H}}_{s-2}(\eta)-{\bf{H}}_{s-3}(\eta))\rVert+\dots+\rVert({\bf{H}}_{s^*+1}(\eta)-{\bf{H}}_{s^*}(\eta)) \rVert\\\nonumber
& \leq \delta^{s-m}\rVert {\bf{F}}_m(\eta,N_1,N_2,\dots,N_m)\rVert+\delta^{s-m-1}\rVert {\bf{F}}_m(\eta,N_1,N_2,\dots,N_m)\rVert\\&\nonumber+\delta^{s-m-2}\rVert {\bf{F}}_m(\eta,N_1,N_2,\dots,N_m)\rVert+\dots+\delta^{s^*-m+1}\rVert {\bf{F}}_m(\eta,N_1,N_2,\dots,N_m) \rVert
\end{align}
\begin{align}\nonumber
\implies \rVert {\bf{H}}_s-{\bf{H}}_{s^*}  \rVert &\leq \left[\delta^{s-m}+\delta^{s-m-1}+\delta^{s-m-2}+\dots+\delta^{s^*-m+1}\right]\rVert {\bf{F}}_m(\eta,N_1,N_2,\dots,N_m) \rVert\\&
 = \dfrac{\delta^{s^*-m+2}}{\delta-1} \rVert {\bf{F}}_m(\eta,N_1,N_2,\dots,N_m) \rVert
\end{align}
\begin{align}
\implies \lim_{s^*\to\infty} \rVert {\bf{H}}_s-{\bf{H}}_{s^*}  \rVert \leq \lim_{s^*\to\infty}\dfrac{\delta^{s^*-m+2}}{\delta-1}\rVert {\bf{F}}_m(\eta,N_1,N_2,\dots,N_m) \rVert \rightarrow 0,\hspace{0.2cm}\because \hspace{0.2cm}0<\delta<1
\end{align}
Thus ${\bf{F^*}}$ is a Cauchy sequence in Hilbert space and hence the result.
\end{proof}
\begin{thm}
If the series (\ref{seriessol}) converges then ${\bf{\overline{F}}}$ will represent the exact solution.
\end{thm}
\begin{proof}
Let ${\bf{\overline{F}}}= \mathlarger{\sum^{s}_{j=0}}{\bf{{F}}}_j(\eta,N_1,N_2,\dots,N_j)$ be the partial sum of series (\ref{seriessol}).\\
Since series (\ref{seriessol}) converges,therefore, $\mathlarger{\lim_{s\to\infty}}{\bf{{F}}}_s(\eta,N_1,N_2,\dots,N_s)$=0.\\
Now, let us take the expansion of ${\bf{{BF}}}$ around imbedding parameter $q$ as
\begin{align}\label{BF}
{\bf{{BF}}}={\bf{B}}_0({\bf{F}}_0)+\sum^{\infty}_{s=1}{\bf{B}}_s({\bf{F}}_0,{\bf{F}}_1,\dots,{\bf{F}}_s)q^s
\end{align}
Substituting assumptions (\ref{ass}) and (\ref{BF}) in equation (\ref{Oham2}), the coefficients of the identical exponents of $q$ on comparison gives:\\
Coefficients of $q^0:$ 
 \begin{align}\label{initial}
 {\bf{A}}({\bf{F}}_0)+{\bf{G}}=0   
 \end{align}
Coefficients of $q^1:$ 
 \begin{align} \label{1st}
  {\bf{A}}({\bf{F}}_1)=N_1{\bf{B}}_0({\bf{F}}_0)
\end{align}
Coefficients of $q^2:$
\begin{align}
{\bf{A}}({\bf{F}}_2)-{\bf{A}}({\bf{F}}_1) =N_2{\bf{B}}_0{\bf{F}}_0+N_1[{\bf{A}}({\bf{F}}_1)+{\bf{B}}_1({\bf{F}}_0,{\bf{F}}_1)]
\end{align} 
Similarly, coefficients of $q^s:$ 
 \begin{align}\label{kthapproximation}
 {\bf{A}}({\bf{F}}_s)-{\bf{A}}({\bf{F}}_{s-1})&=N_s{\bf {B}}_0({\bf{F}}_0)+\sum^{s-1}_{j=1}N_j[{\bf{A}}({\bf{F}}_{s-j})+{\bf{B}}_{s-j}({\bf{F}}_0,\dots,{\bf{ F}}_{s-j})] 
 \end{align}
 Since
 \begin{align} 
{\bf{F}}_0+\sum^{s}_{j=1}{\bf{F}}_j-\sum^{s}_{j=1}{\bf{F}}_{j-1}={\bf{F}}_s,
 \end{align}
 therefore,
 \begin{align} 
\lim_{s\to\infty} \left\{{\bf{F}}_0+\sum^{s}_{j=1}{\bf{F}}_j-\sum^{s}_{j=1}{\bf{F}}_{j-1}\right\}=\lim_{s\to\infty}{\bf{F}}_s=0. 
 \end{align}
  Applying linear operator ${\bf{A}}$, we get
\begin{align} \label{imp}
{\bf{A}}({\bf{F}}_0)+{\bf{A}}\sum^{\infty}_{j=1}{\bf{F}}_j-{\bf{A}}\sum^{\infty}_{j=1}{\bf{F}}_{j-1}=0.
 \end{align}
Utilizing   (\ref{kthapproximation}) in (\ref{imp}), we get 
\begin{align} \nonumber\label{final1}
{\bf{A}}({\bf{F}}_0)+{\bf{A}}\sum^{\infty}_{j=1}{\bf{F}}_j&-{\bf{A}}\sum^{\infty}_{j=1}{\bf{F}}_{j-1}\\
& =\sum^{\infty}_{j=1}\left[\sum^{j-1}_{i=1}N_i[{\bf{A}}{\bf{F}}_{j-i}+{\bf{B}}_{j-i}({\bf{F}}_0,\dots, {\bf{F}}_{j-i})]+N_j{\bf {B}}_0{\bf{F}}_0+{\bf{AF}}_0 \right] =0. 
 \end{align}
Rewriting above equation, 
\begin{align} 
 \sum^{\infty}_{j=1}\left[\sum^{j}_{i=1}N_i[{\bf{A}}{\bf{F}}_{j-i}+{\bf{B}}_{j-i}({\bf{F}}_0,\dots,{\bf{F}}_{j-i})]+{\bf{AF}}_0-N_j{\bf{AF}}_0 \right]=0.
 \end{align}
 If the convergence parameter determined  properly and in chosen in a way that residual becomes zero,  we observe that
 \begin{align}
 {\bf{A}}{\bf{\overline{F}}}+{\bf{B}}{\bf{\overline{F}}}+{\bf {G}}=0, 
 \end{align}
which means that ${\bf{\overline{F}}}$ is the exact solution.
\end{proof}
\subsection{\textbf{Implementation to current problem}} 
Based on previously declared theory, homotopy equations with respect to set of equations (\ref{e1})-(\ref{b1}) are taken as: 
\begin{align}\nonumber\label{homo1}
&(1-q)[F_1''+F_1'-V_r]-(N_1 q+N_2 q^2)\left[F_1'''+(F_1+F_2)F_1''-F_1'(F_1'+F_2')-KF_1'\right.\\&\left.-Fr F_1'^2+V_r\left((2+Fr)V_r+K\right)\right]=0,\\\nonumber
&(1-q)[F_2''+F_2'-V_r]-(N_1 q+N_2 q^2)\left[F_2'''+(F_1+F_2)F_2''-F_2'(F_1'+F_2')-KF_2'\right.\\&\left.-FrF_2'^2+V_r\left((2+Fr)V_r+K\right)\right]=0,\\\nonumber
&(1-q)[\Theta''+\Theta']-(N_1 q+N_2 q^2)[\Theta''+Pr((F_1+F_2)\Theta'+Nb\Theta'\Phi'+Nt\Theta'^2))\\&+{Q_1^*}e^{-\eta}+{Q_2^*}\Theta)]=0, \\
&(1-q)[\Phi''+\Phi']-(N_1 q+N_2 q^2) [\Phi''+Sc(F_1+F_2)\Phi'+\frac{Nt}{Nb}\Theta'']=0
 \end{align}
with modified boundary conditions as 
\begin{equation}\label{homo5}
\left.
\begin{aligned}
 (F_1+F_2, F_1', F_2', \Theta, \Phi)&=\left(0, 1, \lambda, 1, 1 \right) \quad &\text{at} \quad &\eta=0, \\
(F'_1, F_2', \Theta, \Phi)&=(V_r, V_r, 0, 0) \quad &\text{as} \quad &\eta \rightarrow\infty
\end{aligned}
 \right\}.
\end{equation}
Here $(N_1, N_2)$ and $q$ are convergence control and  homotopy parameters respectively. The system of  equations reduces to its original mathematical structure for $q=1$, and for $q=0$ the equations for initial guess are achieved.\\
We will go upto second order of approximation and consider following expansions for $F_1'$, $F_2'$, $\Theta$ and  $\Phi$ as:  
\begin{equation}
\left.
\begin{aligned}\label{heq}
F_1'&=F_{10}'+qF_{11}'+q^2F_{12}', &F_2'&=F_{20}'+qF_{21}'+q^2F_{22}',\\
\Theta&=\Theta_0+q\Theta_1+q^2\Theta_2, &\Phi&=\Phi_0+q\Phi_1+q^2\Phi_2.
\end{aligned}
 \right\}.
\end{equation}
Substitution of these expansions in equations (\ref{homo1})-(\ref{homo5}) and threreafter comparison of coefficients of various identical exponents of $q$ give:\\
{\bf{coefficients of $q^0$:}} 
\begin{align}\label{asol}
 F_{10}''+F_{10}'-V_r=0, \quad F_{20}''+F_{20}'-V_r=0, \quad \Theta_0''+\Theta_0'=0, \quad \Phi_0''+\Phi_0'=0
 \end{align}
with
\begin{align}
&F'_{10}(0)=1, \quad F'_{20}(0)=\lambda, \quad \Theta_0(0)=1, \quad \Phi_0(0)=1, \quad (F_{10}+F_{20})(0)=0,& \nonumber \\
&F'_{10}(\infty)=V_r, \quad F'_{20}(\infty)=V_r, \quad \Theta_0(\infty)=0, \quad \Phi_0(\infty)=0.&
\end{align}
{\bf{coefficients of $q^1$:}}
\begin{align}
F_{11}''+F_{11}'=&(A_1+A_3\eta)e^{-\eta}+A_2e^{-2\eta},
&F_{21}''+F_{21}'=&(B_1+B_3\eta)e^{-\eta}+B_2e^{-2\eta},\\
\Theta_1''+\Theta_1'=&(L_1-2N_1V_rPr\eta)e^{-\eta}+L_2e^{-2\eta},
&\Phi_1''+\Phi_1=&(S_1-2N_1V_rSc\eta)e^{-\eta}+S_2e^{-2\eta}. 
\end{align}
with 
\begin{align}
&F_{11}'(0)=0, \quad F_{21}'(0)=0, \quad \Theta_1(0)=0, \quad \Phi_1(0)=0, \quad (F_{11}+F_{21})(0)=0,& \nonumber \\
&F_{11}'(\infty)=0, \quad F_{21}'(\infty)=0,\quad \Theta_1(\infty)=0, \quad \Phi_1(\infty)=0.& 
\end{align}
{\bf{coefficients of $q^2$:}}
 \begin{align}
 &F_{12}''+F_{12}'=(A_4+A_7\eta+A_9\eta^2-N_1V_rA_3\eta^3)e^{-\eta}+(A_5+A_8\eta+A_{10}\eta^2)e^{-2\eta}+A_6e^{-3\eta},
\end{align}
 \begin{align}
  &F_{22}''+F_{22}'=(B_4+B_7\eta+B_9\eta^2-N_1V_rB_3\eta^3)e^{-\eta}+(B_5+B_8\eta+B_{10}\eta^2)e^{-2\eta}+B_6e^{-3\eta},
 \end{align}
 \begin{align}
 &\Theta_2''+\Theta_2=(L_3+L_6\eta+L_8\eta^2+L_{10}\eta^3)e^{-\eta}+(L_4+L_7\eta+L_{9}\eta^2)e^{-2\eta}+L_5e^{-3\eta},
 \end{align}
 \begin{align}
 &\Phi_2''+\Phi_2=(S_3+S_6\eta+S_8\eta^2+S_{10}\eta^3)e^{-\eta}+(S_4+S_7\eta+S_{9}\eta^2)e^{-2\eta}+S_5e^{-3\eta}
 \end{align}
with
\begin{align}\label{bsol}
&F_{12}'(0)=0, \quad F_{22}'(0)=0, \quad \Theta_2(0)=0, \quad \Phi_2(0)=0, \quad (F_{12}+F_{22})(0)=0,& \nonumber \\
&F_{12}'(\infty)=0, \quad F_{22}'(\infty)=0,\quad \Theta_2(\infty)=0, \quad \Phi_2(\infty)=0.&
\end{align}
The solutions to above set of equations (\ref{asol})-(\ref{bsol}) is obtained as:
\begin{align}
F_{10}'&=(1-V_r)e^{-\eta}+V_r,&\\
F_{20}'&=(\lambda_1-V_r)e^{-\eta}+V_r,&\\
F_{11}'&=\left(A_2+A_1\eta+\dfrac{A_3}{2}\eta^2\right)e^{-\eta}-A_2e^{-2\eta},&\\
F_{21}'&=\left(B_2+B_1\eta+\dfrac{B_3}{2}\eta^2\right)e^{-\eta}-B_2e^{-2\eta},&\\\nonumber
F_{12}'&=\left(A_5+A_8+2A_{10}+\dfrac{A_6}{2}+A_4\eta+\dfrac{A_7}{2}\eta^2+\dfrac{A_9}{3}\eta^3-\dfrac{N_1V_rA_3}{4}\eta^4\right)e^{-\eta}\\&-[A_5+(A_8+2A_{10})(1+\eta)+A_{10}\eta^2]e^{-2\eta}-\dfrac{A_6}{2}e^{-3\eta},&\\\nonumber
F_{22}'&=\left(B_5+B_8+2B_{10}+\dfrac{B_6}{2}+B_4\eta+\dfrac{B_7}{2}\eta^2+\dfrac{B_9}{3}\eta^3-\dfrac{N_1V_rB_3}{4}\eta^4\right)e^{-\eta}\\&-[B_5+(B_8+2B_{10})(1+\eta)+B_{10}\eta^2]e^{-2\eta}-\dfrac{B_6}{2}e^{-3\eta},&\\
\Theta_0&=e^{-\eta},&\\
\Theta_1&=\left[-\dfrac{L_2}{2}+(2N_1PrV_r-L_1)\eta+N_1PrV_r\eta^2\right]e^{-\eta}+\dfrac{L_2}{2}e^{-2\eta},&\\\nonumber
\Theta_2&=-\left[\dfrac{L_4}{2}+\dfrac{3L_7}{4}+\dfrac{7L_9}{4}+\dfrac{L_5}{6}+(L_3+L_6+2L_8+6L_{10})\eta+\left(\dfrac{L_6}{2}+L_8+3L_{10}\right)\eta^2\right.\\\nonumber&\left.+\left(\dfrac{L_8}{3}+L_{10}\right)\eta^3+\dfrac{L_{10}}{4}\eta^4\right]e^{-\eta}+\left[\dfrac{L_4}{2}+\dfrac{3L_7}{4}+\dfrac{7L_9}{4}+\left(\dfrac{L_7}{2}+\dfrac{3L_9}{2}\right)\eta+\dfrac{L_9}{2}\eta^2\right]e^{-2\eta}\\&+\dfrac{L_5}{6}e^{-3\eta},&\\
\Phi_0&=e^{-\eta},&\\
\Phi_1&=\left[-\dfrac{S_2}{2}+(2N_1PrV_r-S_1)\eta+N_1PrV_r\eta^2\right]e^{-\eta}+\dfrac{S_2}{2}e^{-2\eta},&\\\nonumber
\Phi_2&=-\left[\dfrac{S_4}{2}+\dfrac{3S_7}{4}+\dfrac{7S_9}{4}+\dfrac{S_5}{6}+(S_3+S_6+2S_8+6S_{10})\eta+\left(\dfrac{S_6}{2}+S_8+3S_{10}\right)\eta^2\right.\\\nonumber&\left.+\left(\dfrac{S_8}{3}+S_{10}\right)\eta^3+\dfrac{S_{10}}{4}\eta^4\right]e^{-\eta}+\left[\dfrac{S_4}{2}+\dfrac{3S_7}{4}+\dfrac{7S_9}{4}+\left(\dfrac{S_7}{2}+\dfrac{3S_9}{2}\right)\eta+\dfrac{S_9}{2}\eta^2\right]e^{-2\eta}\\&+\dfrac{S_5}{6}e^{-3\eta}.
 \end{align}
These solutions when substituted in equations (\ref{heq}) will provide the final approximate solution (after letting $q=1$) as:
 \begin{align}
 F_1'=\sum^{2}_{j=0}F_{1j}',\quad \quad  F_2'=\sum^{2}_{j=0}F_{2j}', \quad \quad \Theta=\sum^{2}_{j=0}\Theta_j,\quad \quad \Phi=\sum^{2}_{j=0}\Phi_j.
 \end{align}
Different coefficients which appeared above are not reported here to protect the space.
\section{Physical exploration of results}
The convergent solutions determined in previous section are physically explored for different parameters like $\lambda$ (velocity ratio due to stretching), $V_r$ (velocity ratio due to stagnation and stretching), $K$ (permeability parameter), $Fr$ (Forchheimer parameter), $Q_1^*$ or $Q_2^*$ (heat source/sink), $Sc$ (Schmidt number), $Nb$ (Brownian motion) and  $Nt$ (thermophoresis) through tabulation and plotting. Convergence controlling parameters (determined through OHAM) which keep residual functions at their minimum are established as $N_1 = 0.04395191567136332$, $N_2  = -0.21630125455453278$.\\

Figures (\ref{fdeshwithVr}) and (\ref{gdeshwithVr}) explain the significance of velocity ratio $V_r$ on velocity components $F_1'$ and $F_2'$. When $V_r>1$, profiles of $F_1'$ and $F_2'$ are parabolically similar and approach boundary conditions (at free stream) in asymptotic sense. The similar structure of $F_1'$ and $F_2'$ is maintained by the presence of strong similar stagnation forces in two lateral directions.  However when $V_r\le1$, $F_1'$ profiles shape are inverted because of the dominance of different but vigorous stretching forces in two directions in comparison to weak free stream flow strength. \\

Figures (\ref{fgdeshwithl1})-(\ref{fgdeshwithl1Vr}) bring out the influences of stretching velocity ratio parameter $\lambda$ on $F_1'$ and $F_2'$ (nanofluid velocities) by taking two values of $V_r$ $(=0.5, 1.5).$ When $V_r=0.5$, both components of nanofluid velocity $(F_1', F_2')$ are lifted up with the increasing parameter $\lambda$ in a region nearer to sheet. However away from the sheet, this pattern is reversed for component $F_2'$ (Fig. \ref{fgdeshwithl1}). Here, primary velocity is accelerated with stretching ratio in the entire boundary layer due to the dominance of primary stretching. Fig. \ref{fgdeshwithl1Vr} also portrays increasing profiles of $F_1'$ and $F_2'$ with $\lambda$ near to sheet when $V_r=1.5$, but trend is reversed for both components $F_1'$ and $F_2'$ away from the sheet. \\

Figures (\ref{fgdeshwithFr})-(\ref{fgdeshwithKVr}) illustrate the variability of $F_1'$ and $F_2'$ with Forchheimer factor $Fr$ and permeability parameter $K$ for two case of $V_r$. In the presence of vigorous free stream strengths (i.e. $V_r=1.5$), nanofluid velocities in boundary layer are reduced with $Fr$. This is natural to occur here because form-drag coefficient due to solid obstacles is comparable to surface drag due to friction. Since $Fr$ has direct association with form-drag coefficient therefore, as $Fr$ is augmented, velocity should be depressed.\\

 For weakened free stream strengths in relation to stretching forces (i.e. $V_r=0.5$), $F_2'$ is hampered with $Fr$ but $F_1'$ is augmented with it. In this case, free stream velocity is lower than primary stretching velocity which itself is higher than stretching forces in secondary direction. Due to these nanofluid velocity relations, opposite influence of $Fr$ is noticed on $F_1'$ (that is,  $F_1'$ is not a decreasing function of $Fr$), but $Fr$ positively influence $F_2'$. Therefore $F_1'$ should be raised but $F_2'$ should diminish with $Fr$. Same is happening in present case. Figures (\ref{fgdeshwithK})-(\ref{fgdeshwithKVr}) depict the similar trends with $K$ as those with $Fr$ but here overshoot points are marginally smaller.  \\
Figures (\ref{thetawithNt})-(\ref{thetawithNbVr}) illuminate the significance of $Nt$ and $Nb$ on the nanoparticle temperature $\Theta (\eta)$ variations under heat source $(Q_1^*>0$, $Q_2^*>0)$ and heat sink $(Q_1^*<0$, $Q_2^*<0)$ considerations in the flow field. Temperature of nanoparticle volume fraction is enhanced with $Nb$ and $Nt$ for all the restrictions of heat source/sink and stagnation flow strengths. Here negative temperature profiles are detected as per inverted Boltzmann transport. Similar profiles have also been seen in a particular three dimensional transport of nanofluid analyzed by Kumar and Sood \cite{kumarsood2017} .  \\

Figures (\ref{phiwithNt})-(\ref{phiwithSc}) exhibit the influence of $Nt$, $Nb$ and $Sc$ on nanoparticle concentration distribution. Distribution of nanoparticle concentration is lifted up with $Nt$ and $Sc$ but depressed with $Nb$ under different ratios of free stream and stretching velocities. since higher values of $Sc$ indicates the higher momentum diffusion than mass diffusion, therefore, results are physically realizable.\\
 
In tables (\ref{tb1}) and (\ref{tb2}) effects of permeability parameter ($K$), Forchheimer factor ($Fr$), stretching velocity ratio ($\lambda$) and heat source/sink parameter ($Q_1^*$ and $Q_2^*$) on local skin-friction rates $(F_1''(0), F_2''(0))$, local Nusselt number $(\Theta'(0))$ and local Sherwood number $(\Phi'(0))$ are examined under different velocity ratios (stagnation to stretching)  considering $V_r<1$ and $V_r>1$. When $V_r<1$, $K$, $Fr$ and $\lambda$ hamper both skin-friction coefficients, but no effect of heat source has been seen on these coefficients. These observations are physically correct because stretching forces accelerate the nanofluid velocity which means skin-friction factor should reduce. Table (\ref{tb1}) also show that Nusselt number is curtailed with $K$, $Fr$ and $\lambda$, but opposite influence on Sherwood number has been realized with these parameters. Heat source also cutback $\Theta'(0)$ and $\Phi'(0)$.\\

On other side when $V_r>1$, magnitude of  $F_1''(0)$ is raised with $K$ and $Fr$, but $F_2''(0)$ is declined with these parameters (Table (\ref{tb2})). However both components of skin-friction are diminished with $\lambda$. Further, $\Theta'(0)$ is incremented with $K$ but $\Phi'(0)$ declined with it, and reverse effect is detected with $\lambda$. Both heat and mass transport rates are depressed with heat source, and realize opposite effect with $Fr$.\\
 
 Tables (\ref{tb3})-(\ref{tb4}) highlight the effects of $Nt$, $Nb$ and $Sc$ on the magnitude of local numbers $\Theta'(0)$ and $\Phi'(0)$ considering different velocity ratio aspects ($V_r<1$ or $V_r>1$) and heat source/sink impact. It is inferred that both $|\Theta'(0)|$ and $|\Phi'(0)|$ are decreasing functions of $Nt$ and $Sc$ under heat sink environment, however local Nusselt number is elevated and local Sherwood number is depressed under heat absorption $(Q_1^*=Q_2^*<0)$. Similar nature of $|\Theta'(0)|$ and $|\Phi'(0)|$ has also been observed for heat generations $(Q_1^*=Q_2^*>0)$. The pattern of these rates is invariant under velocity ratio ($V_r$) variations. The interesting part is the higher heat transfer rate for $V_r>1$ than other case when $V_r<1$, but reverse is the phenomenon for mass transfer rate. Thus here we interpret that a heat sink device is required to enhance the heat transfer rate.
\section{Conclusion}
Three dimensional stagnation-point flow of nanofluid over a  stretching sheet has been examined via OHAM considering Darcy-Forchheimer medium and inconsistent heat source/sink. Mathematical theory of OHAM is revisited, and a unique presentation is provided for the handling of vector differential equation set up. Necessary theorems for convergence of results are also given. Based on achieved results, the concluding remarks are: 
\begin{itemize}
\item{Under weak free stream velocity, Forchheimer factor negatively influence primary nanofluid velocity. But for strong free stream velocity, primary and secondary nanofluid velocities behave in a positive way as per Forchheimer factor physics.}
\item{For the case $V_r>1$, nanofluid velocity decreased with $K$ and $Fr$. Same trend is followed by the velocity field in case of $V_r<1$ except that the primary velocity increases with $K$ and $Fr$.}
\item{Double behavior of nanofluid velocities is noted with stretching ratio parameter when stagnation flow strength overcome surface stretching forces.}
\item{Heat transport rate is higher when stagnation flow strength dominates stretching forces. Reverse is the case with mass transport rate.}
\item{Heat transport rate is depressed with permeability parameter and Forchheimer coefficient, and this rate can be enhanced by introducing heat sink in the flow field.}
\end{itemize}
\clearpage
\begin{figure}[!t]
\makebox[\textwidth][c]{\includegraphics[width=1.5\textwidth]{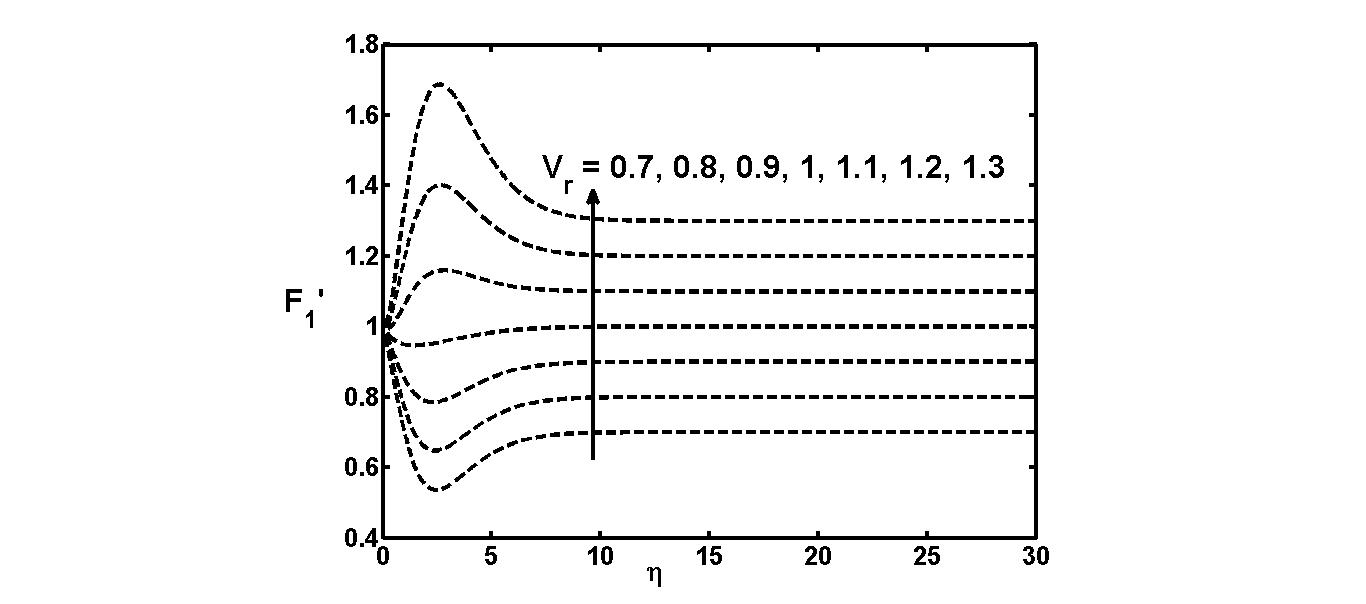}}%
\caption{Variability of $F_1'$ with $V_r$  for $\lambda=0.1$, $Fr=1$,  $K=1$, $Q_1^*=Q_2^*=1$, $Nt=0.2$, $Nb=0.2$, $ Sc=1$, $Pr=6.2$.}\label{fdeshwithVr}
\end{figure}
\begin{figure}[!t]
\makebox[\textwidth][c]{\includegraphics[width=1.5\textwidth]{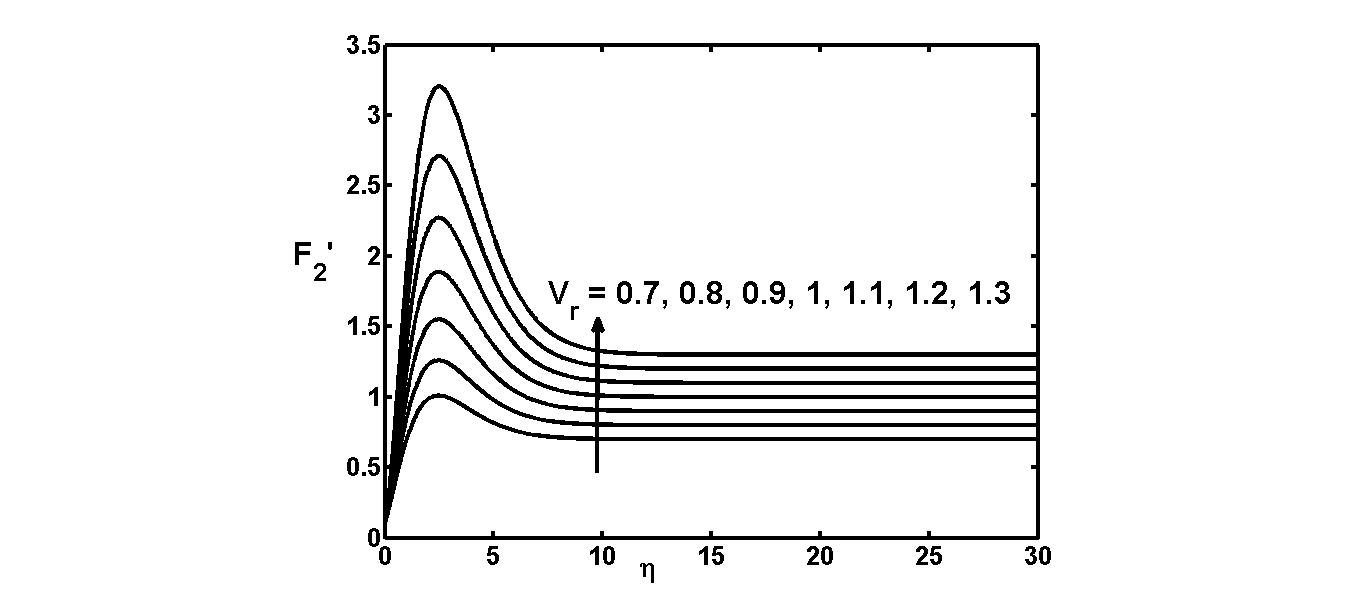}}%
\caption{Variability of $F_2'$ with $V_r$  for $\lambda=0.1$, $Fr=1$,  $K=1$, $Q_1^*=Q_2^*=1$, $Nt=0.2$, $Nb=0.2$, $ Sc=1$, $Pr=6.2$.}\label{gdeshwithVr}
\end{figure}
\begin{figure}[!t]
\makebox[\textwidth][c]{\includegraphics[width=1.5\textwidth]{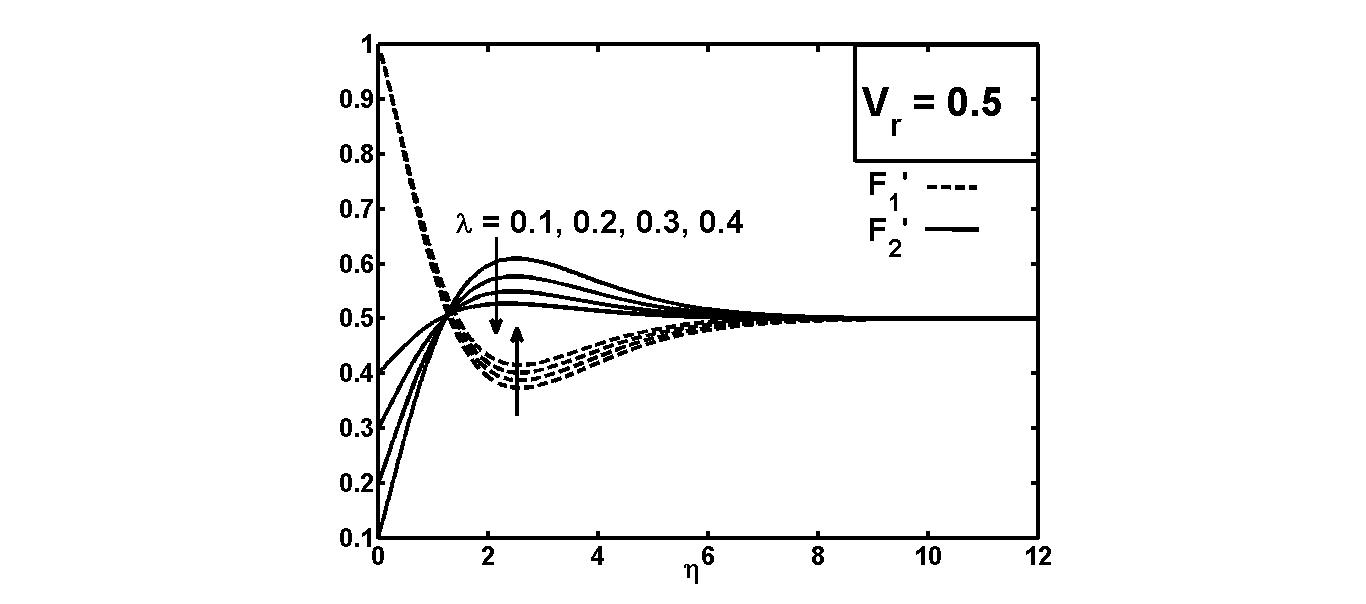}}%
\caption{Variability of $F_1'$ and $F_2'$ with $\lambda$  for $Fr=1$,  $K=1$, $Q_1^*=Q_2^*=1$, $Nt=0.2$, $Nb=0.2$, $ Sc=1$, $Pr=6.2$.}\label{fgdeshwithl1}
\end{figure}
\begin{figure}[!t]
\makebox[\textwidth][c]{\includegraphics[width=1.5\textwidth]{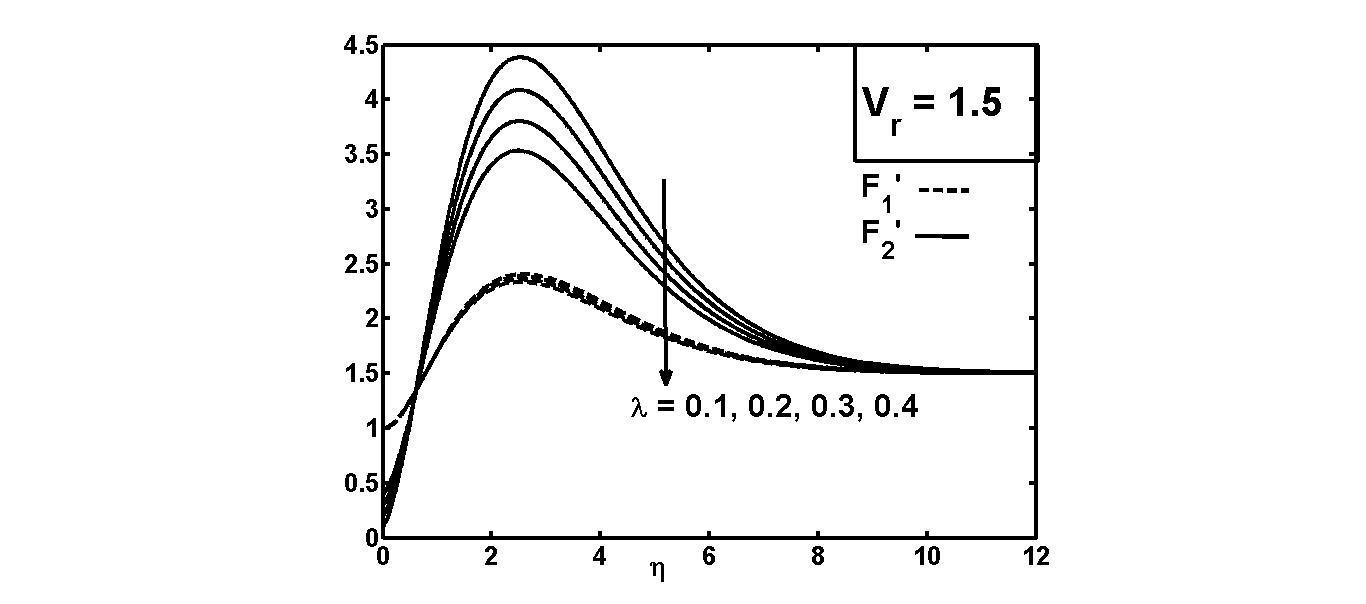}}%
\caption{Variability of $F_1'$ and $F_2'$ with $\lambda$  for $Fr=1$,  $K=1$, $Q_1^*=Q_2^*=1$, $Nt=0.2$, $Nb=0.2$, $ Sc=1$, $Pr=6.2$.}\label{fgdeshwithl1Vr}
\end{figure}
\begin{figure}[!t]
\makebox[\textwidth][c]{\includegraphics[width=1.5\textwidth]{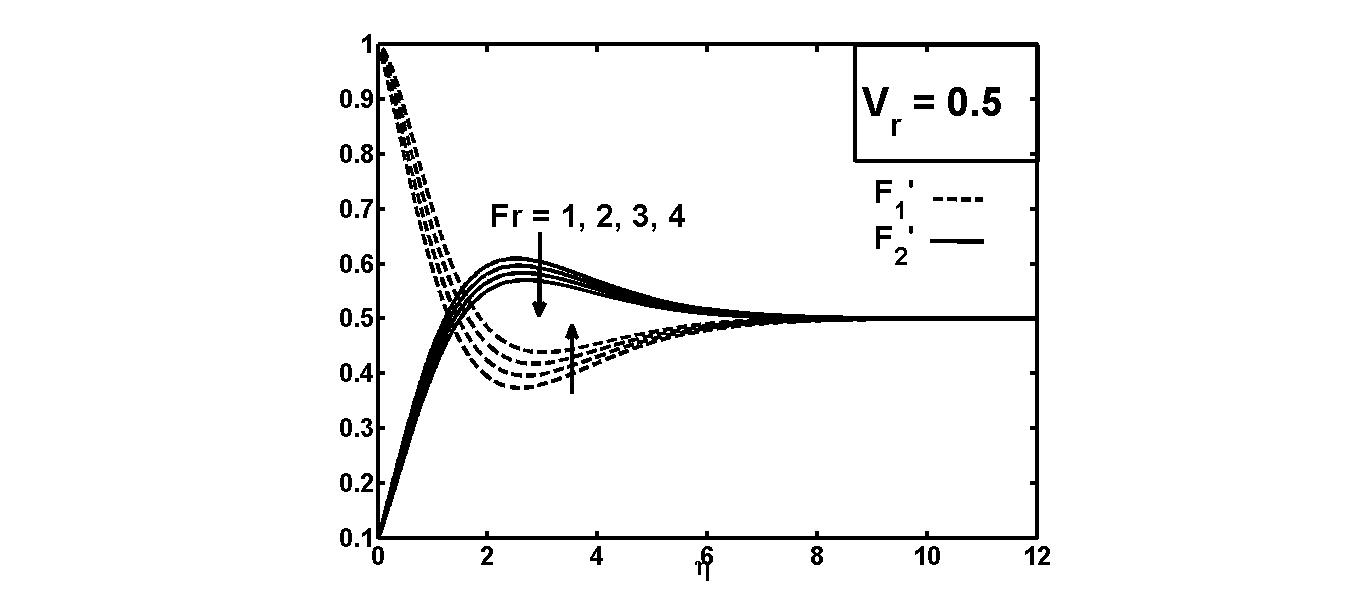}}%
\caption{Variability of $F_1'$ and $F_2'$ with $Fr$  for $\lambda=0.1$,  $K=1$, $Q_1^*=Q_2^*=1$, $Nt=0.2$, $Nb=0.2$, $ Sc=1$, $Pr=6.2$.}\label{fgdeshwithFr}
\end{figure}
\begin{figure}[!t]
\makebox[\textwidth][c]{\includegraphics[width=1.5\textwidth]{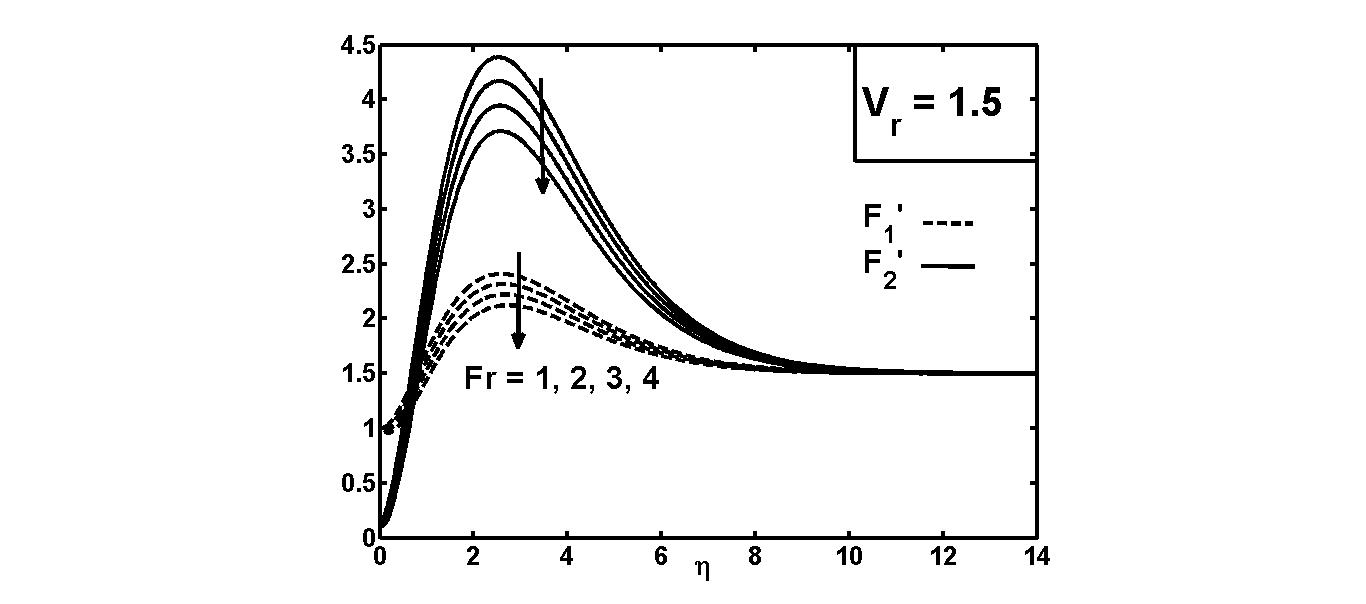}}%
\caption{Variability of $F_1'$ and $F_2'$ with $Fr$  for $\lambda=0.1$,  $K=1$, $Q_1^*=Q_2^*=1$, $Nt=0.2$, $Nb=0.2$, $ Sc=1$, $Pr=6.2$.}\label{fgdeshwithFrVr}
\end{figure}
\begin{figure}[!t]
\makebox[\textwidth][c]{\includegraphics[width=1.5\textwidth]{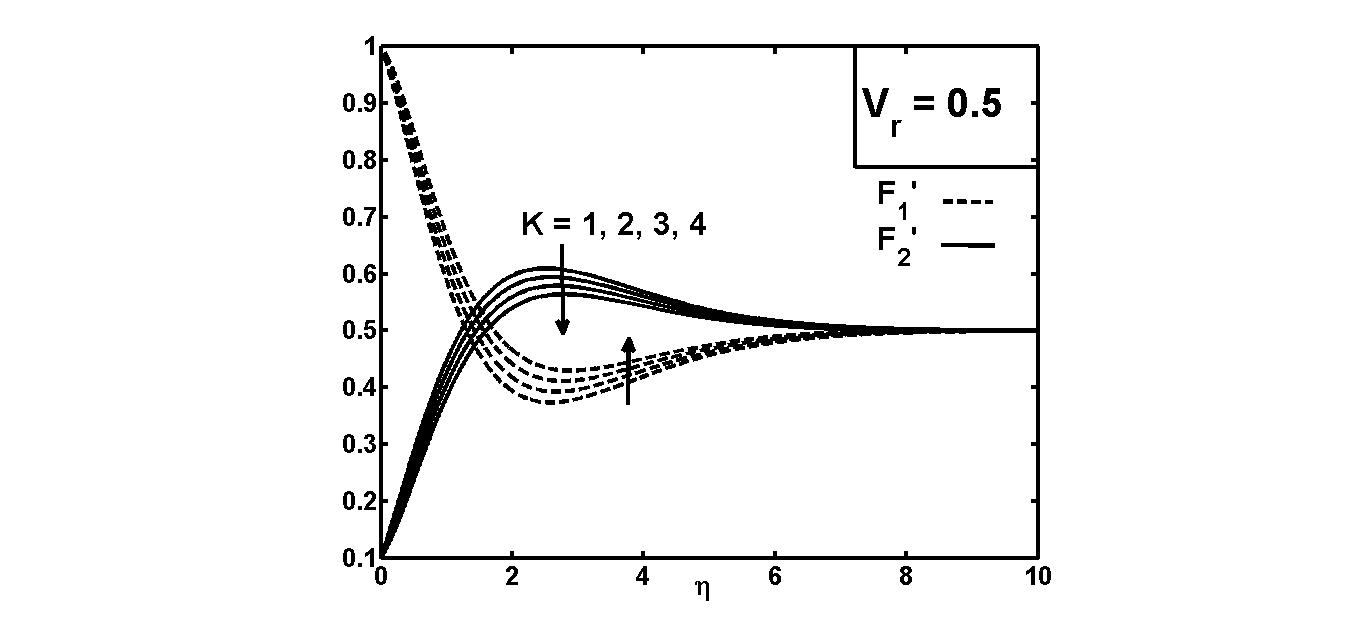}}%
\caption{Variability of $F_1'$ and $F_2'$ with $K$  for $\lambda=0.1$, $Fr=1$,  $Q_1^*=Q_2^*=1$, $Nt=0.2$, $Nb=0.2$, $ Sc=1$, $Pr=6.2$.}\label{fgdeshwithK}
\end{figure}
\begin{figure}[!t]
\makebox[\textwidth][c]{\includegraphics[width=1.5\textwidth]{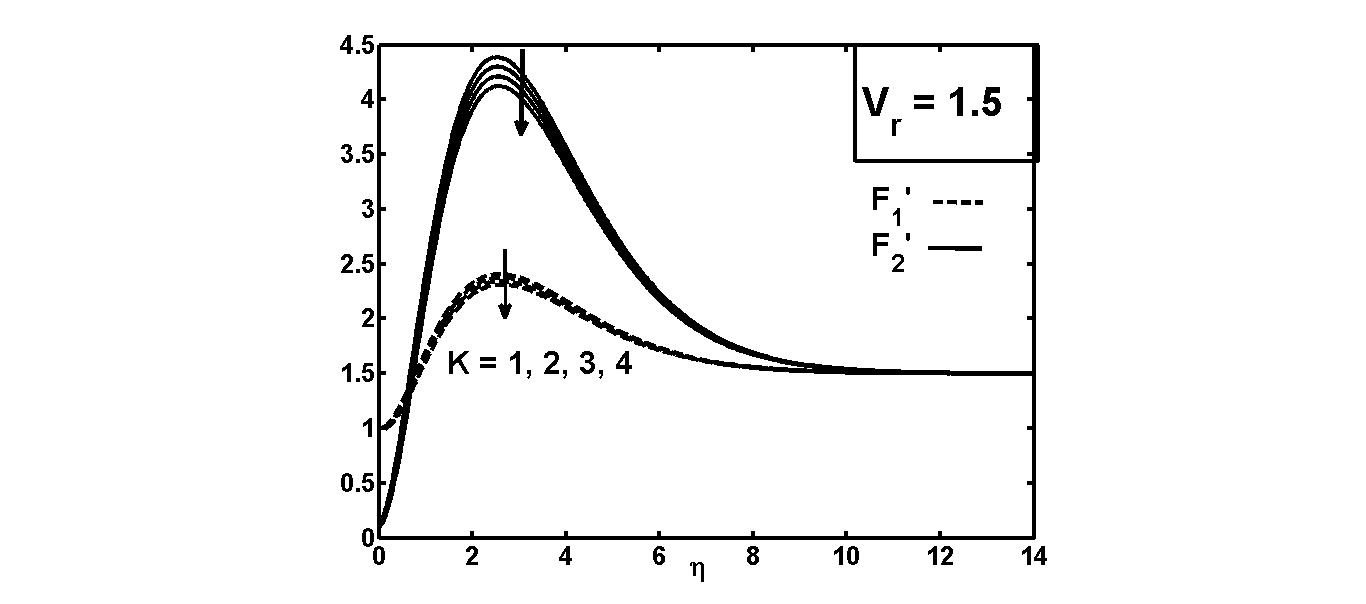}}%
\caption{Variability of $F_1'$ and $F_2'$ with $K$  for $\lambda=0.1$, $Fr=1$, $Q_1^*=Q_2^*=1$, $Nt=0.2$, $Nb=0.2$, $ Sc=1$, $Pr=6.2$.}\label{fgdeshwithKVr}
\end{figure}
\begin{figure}[!t]
\makebox[\textwidth][c]{\includegraphics[width=1.5\textwidth]{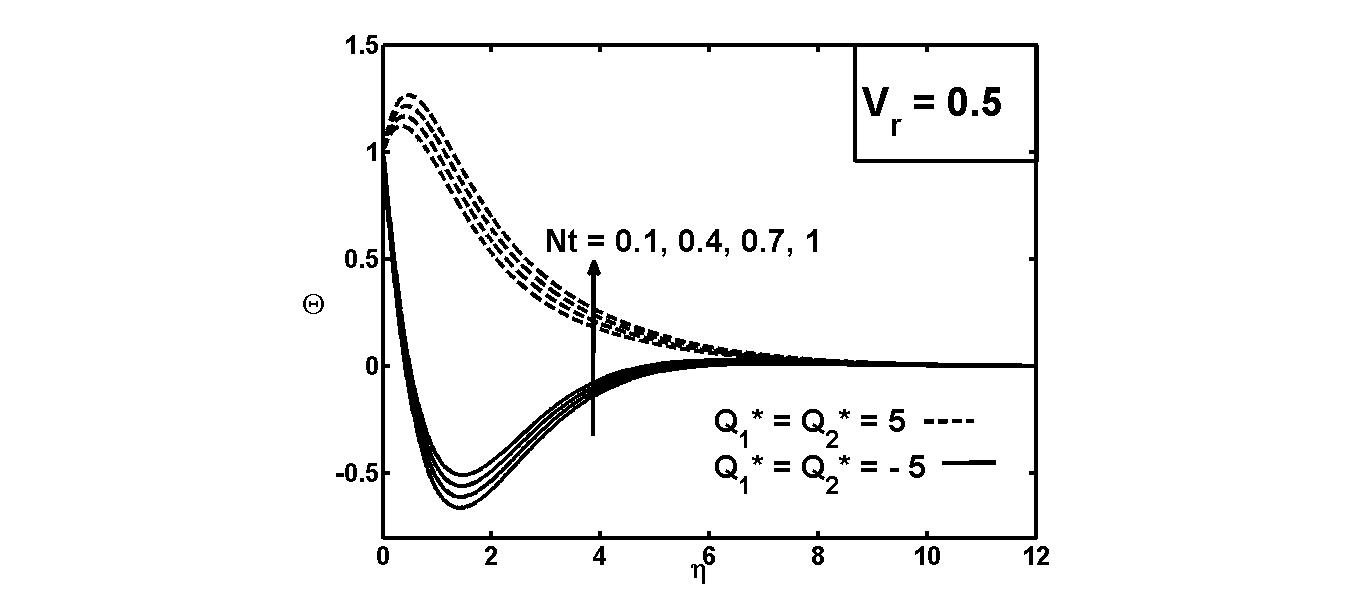}}%
\caption{Variability of $\Theta$ with $Nt$  for for $\lambda=0.1$, $Fr=1$,  $K=1$, $Nb=0.2$, $ Sc=1$, $Pr=6.2$.}\label{thetawithNt}
\end{figure}
\begin{figure}[!t]
\makebox[\textwidth][c]{\includegraphics[width=1.5\textwidth]{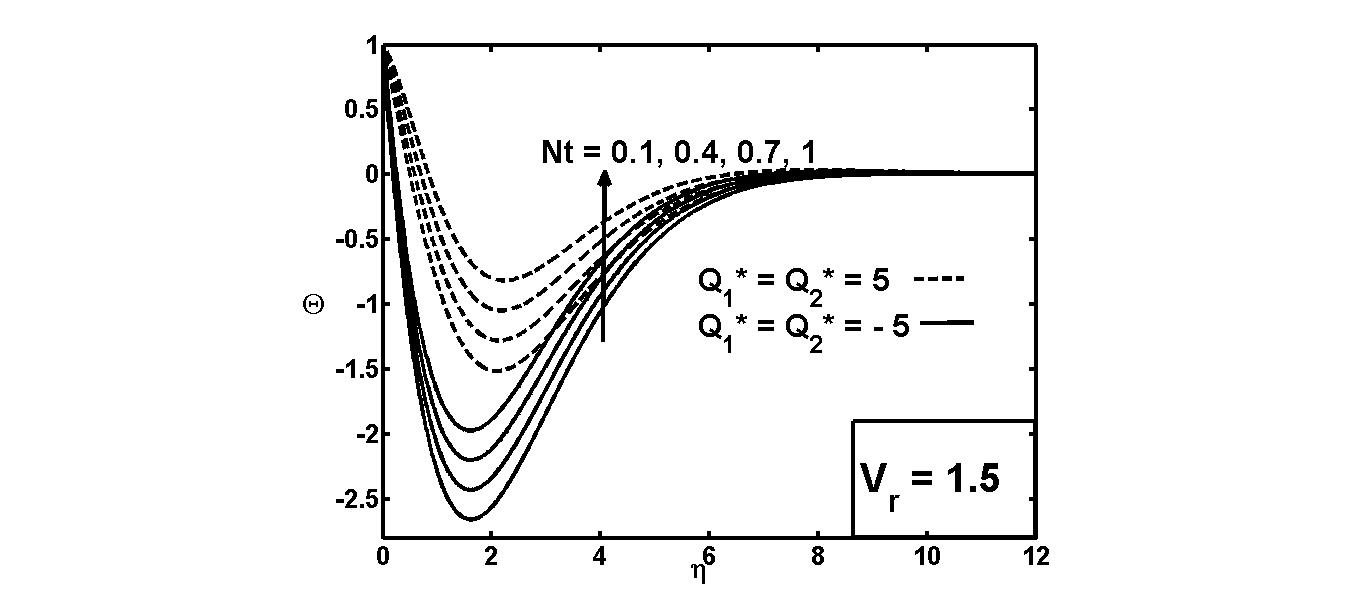}}%
\caption{Variability of $\Theta$ with $Nt$  for $\lambda=0.1$, $Fr=1$,  $K=1$, $Nb=0.2$, $ Sc=1$, $Pr=6.2$.}\label{thetawithNtVr}
\end{figure}
\begin{figure}[!t]
\makebox[\textwidth][c]{\includegraphics[width=1.5\textwidth]{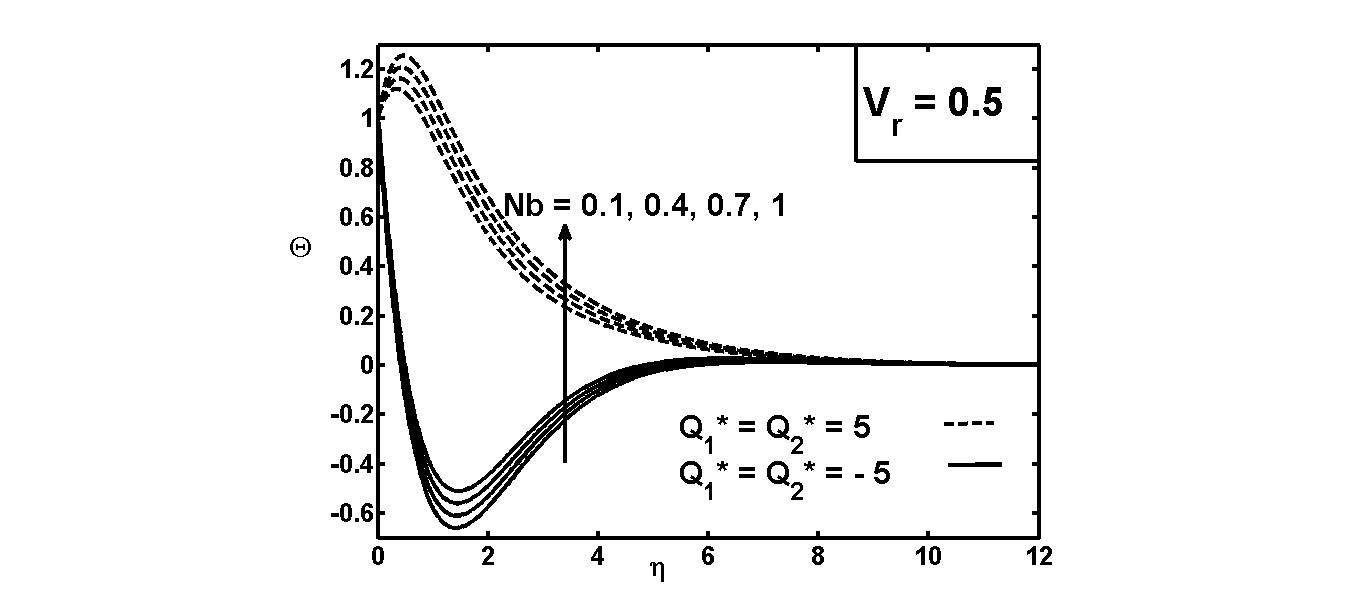}}%
\caption{Variability of $\Theta$ with $Nb$  for $\lambda=0.1$, $Fr=1$,  $K=1$, $Nt=0.2$, $ Sc=1$, $Pr=6.2$.}\label{thetawithNb}
\end{figure}
\begin{figure}[!t]
\makebox[\textwidth][c]{\includegraphics[width=1.5\textwidth]{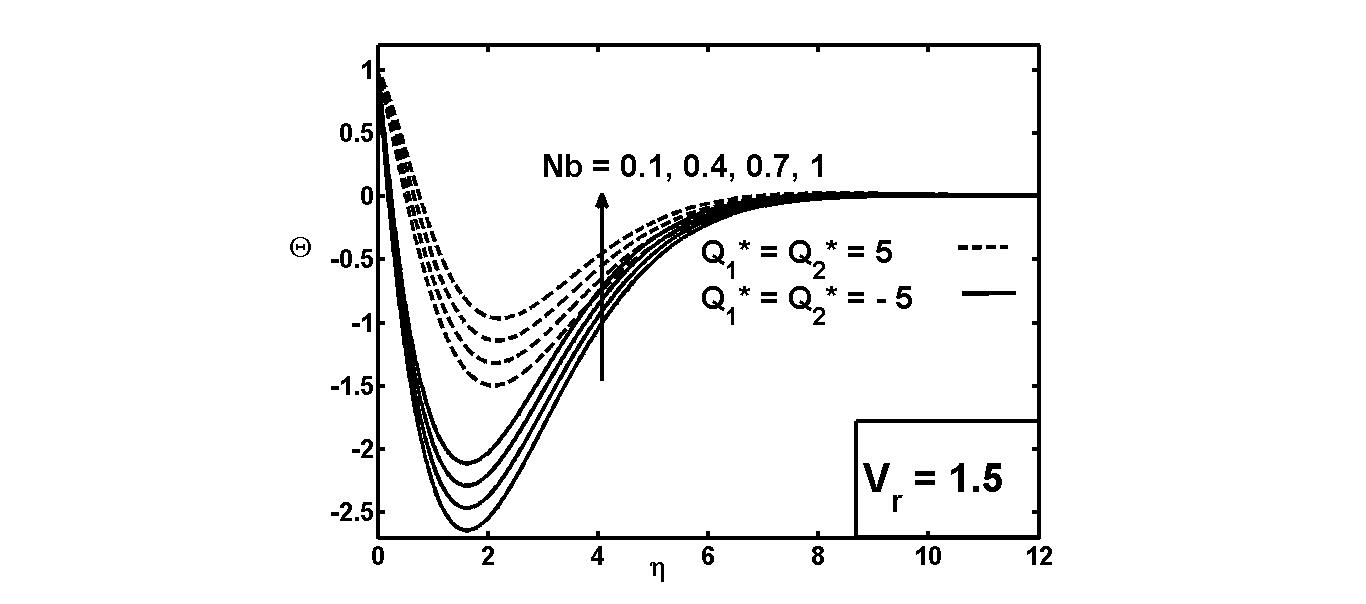}}%
\caption{Variability of $\Theta$ with $Nb$  for $\lambda=0.1$, $Fr=1$,  $K=1$, $Nt=0.2$, $ Sc=1$, $Pr=6.2$.}\label{thetawithNbVr}
\end{figure}
\begin{figure}[!t]
\makebox[\textwidth][c]{\includegraphics[width=1.5\textwidth]{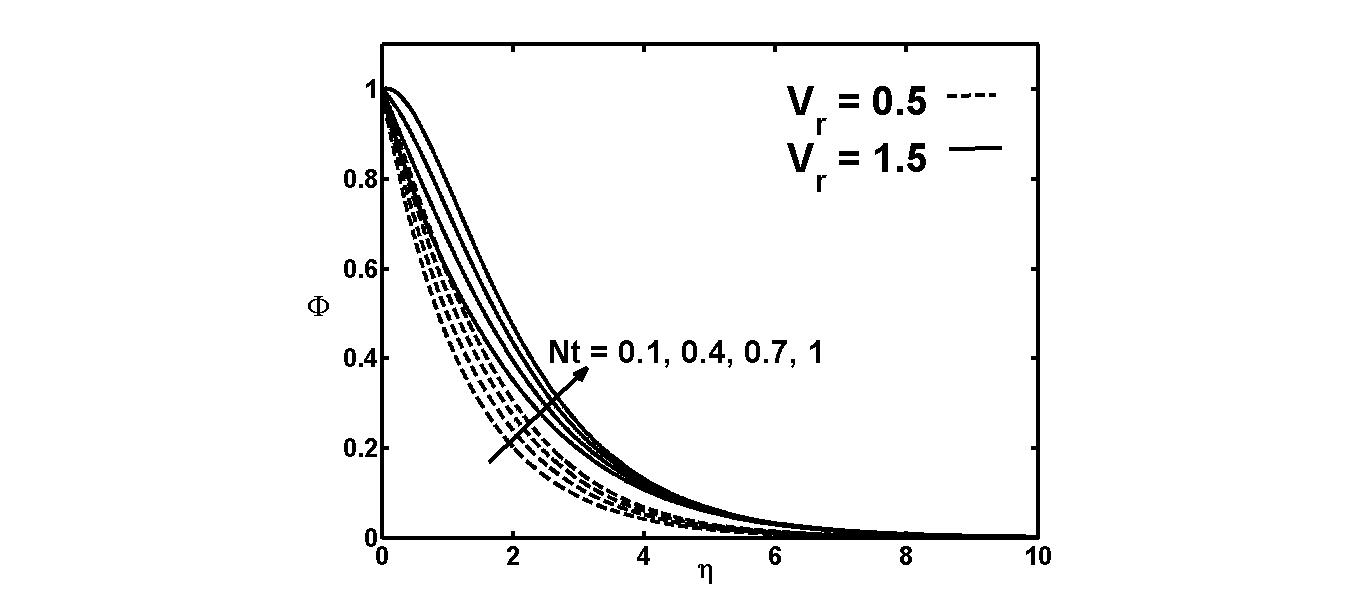}}%
\caption{Variability of $\Phi$ with $Nt$  for $\lambda=0.1$, $Fr=1$,  $K=1$, $Q_1^*=Q_2^*=1$, $Nb=0.2$, $ Sc=1$, $Pr=6.2$.}\label{phiwithNt}
\end{figure}
\begin{figure}[!t]
\makebox[\textwidth][c]{\includegraphics[width=1.5\textwidth]{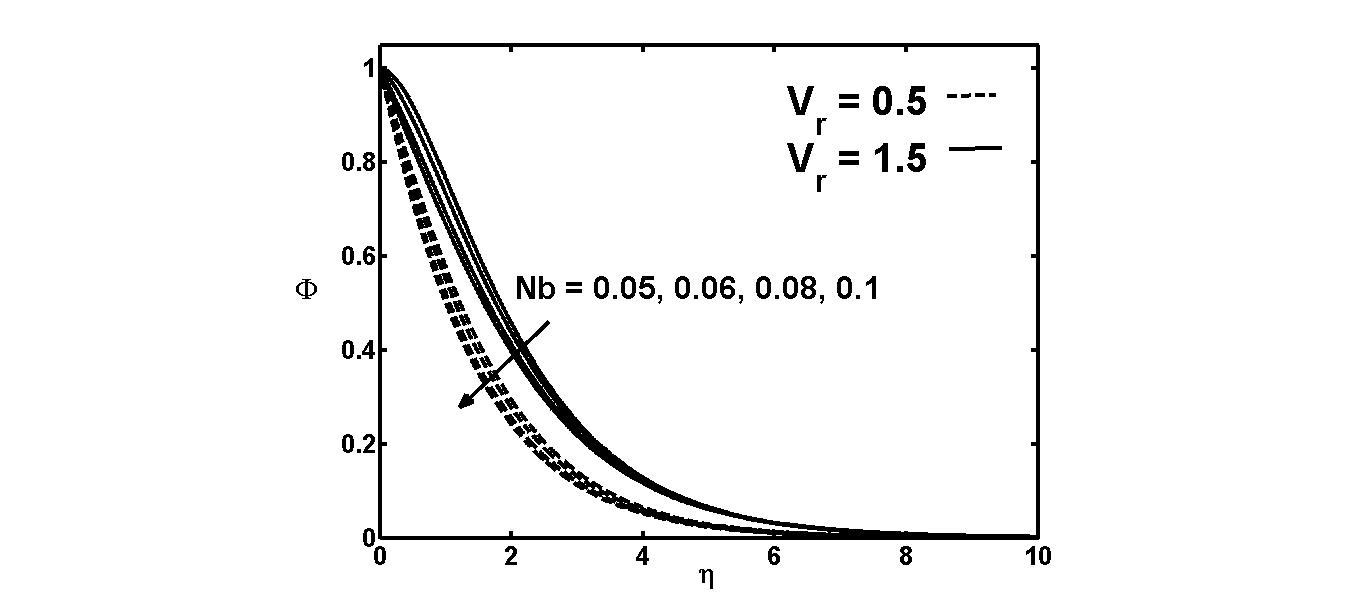}}%
\caption{Variability of $\Phi$ with $Nb$  for $\lambda=0.1$, $Fr=1$,  $K=1$, $Q_1^*=Q_2^*=1$, $Nt=0.2$, $ Sc=1$, $Pr=6.2$.}\label{phiwithNb}
\end{figure}
\begin{figure}[!t]
\makebox[\textwidth][c]{\includegraphics[width=1.5\textwidth]{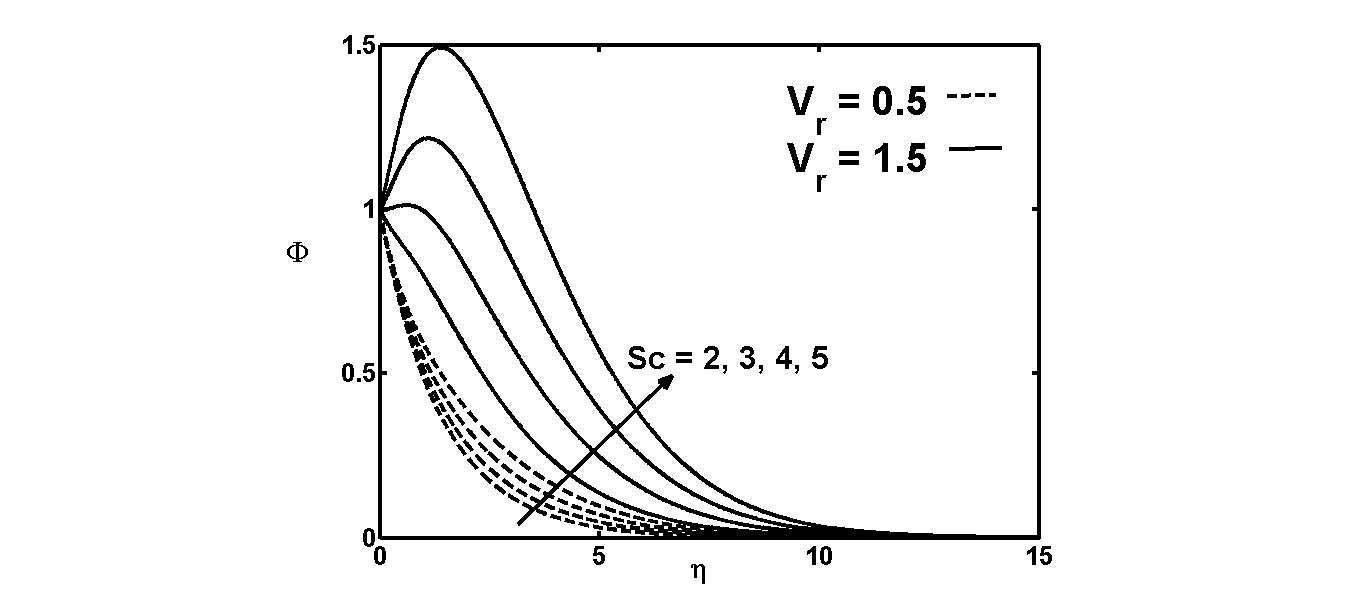}}%
\caption{Variability of $\Phi$ with $Sc$  for $\lambda=0.1$, $Fr=1$,  $K=1$, $Q_1^*=Q_2^*=1$, $Nt=0.2$, $Nb=0.2$, $Pr=6.2$.}\label{phiwithSc}
\end{figure}

\begin{table*}[!]
\begin{center}
\caption{ Values of $(F_1''(0), F_2''(0))$, $|\Theta'(0)|$ and $|\Phi'(0)|$ against $\lambda=0.1$, $Nt=0.2$, $Nb=0.2$, $ Sc=1$, $Pr=6.2$ when $V_r<1$.}\label{tb1}
\vspace* {0.1cm}
\begin{tabular} 
{ p{0.3in} p{0.3in} p{0.3in}  p{0.5in}   p{0.7in}   p{0.7in}  p{0.7in}  p{0.7in}  p{0.7in}  p{0.7in} }\hline
 $K$&$Fr$&$\lambda$&\hspace{-0.4cm}$Q_1^*=Q_2^*$ & \hspace{0cm}$F_1''(0)$&$F_2''(0)$& $|\Theta'(0)|$ &\hspace{0.1cm}$|\Phi'(0)|$ \\  \hline 
 
 1& 1&0.1&1 &-0.342941 &0.334565&0.790426&0.788254 \\
  2&&&&  -0.276810& 0.281661&0.788629&0.788302     \\
  3&&&& -0.210680&0.228757&0.786832 &0.788350   \\
  4&&&& -0.144549&0.175852&0.785036&0.788399        \\
  
  1& 1&0.1& 1&-0.342941 &0.334565&0.790426&0.788254 \\
  &2&&&  -0.243262&0.303132&0.789447&0.789094   \\
  &3&&&-0.143583&0.271698&0.788469 &0.789935   \\
  &4&&&-0.043905&0.240265&0.787491&0.790775    \\
  
  1&1&0.1& 1&-0.342941 &0.334565&0.790426&0.788254 \\
  &&0.2&&  -0.329715&0.251197&0.782007&0.800461   \\
  &&0.3&&-0.316489&0.173157&0.775620&0.812801     \\
  &&0.4&&-0.303262&0.100447&0.771265&0.825273  \\
  
  1&1&0.1&0.5 &-0.342941 &0.334565&0.997969&0.788452\\
  &&&1& -0.342941&0.334565&0.790426&0.788254   \\
  &&&1.5& -0.342941&0.334565&0.582798&0.788056  \\
  &&&2&-0.342941&0.334565&0.375085&0.787858 \\\hline
  
\end{tabular}
\end{center}
\end{table*} 
\begin{table*}[!]
\begin{center}
\caption{Values of $(F_1''(0), F_2''(0))$, $|\Theta'(0)|$ and $|\Phi'(0)|$ against $\lambda=0.1$, $Nt=0.2$, $Nb=0.2$, $ Sc=1$, $Pr=6.2$ when $V_r>1$.}\label{tb2}
\vspace* {0.1cm}
\begin{tabular} 
{ p{0.3in} p{0.3in} p{0.3in}  p{0.5in}   p{0.7in}   p{0.7in}  p{0.7in}  p{0.7in}  p{0.7in}  p{0.7in} }\hline
$K$&$Fr$&$\lambda$&\hspace{-0.4cm}$Q_1^*=Q_2^*$ & \hspace{0cm}$F_1''(0)$&$F_2''(0)$& $|\Theta'(0)|$ &\hspace{0.1cm}$|\Phi'(0)|$ \\  \hline 

 1& 1&0.1&1 &-0.048603 &0.711494&3.478013&0.520178 \\
  2&&&&-0.114733& 0.526329&3.512147&0.519261    \\
  3&&&&-0.180864&0.341163&3.546281&0.518343 \\
  4&&&& -0.246994&0.155998&3.580416&0.517426      \\

  1& 1&0.1&1 &-0.048603 &0.711494&3.478013&0.520178 \\
  &2&&& -0.213446& 0.419016&3.584827&0.521695   \\
  &3&&& -0.378289&0.126538&3.691642&0.523211 \\
  &4&&&-0.543132&-0.165941&3.798457&0.524728     \\

  1&1 &0.1& 1&-0.048603&0.711494&3.478013&0.520178 \\
  &&0.2&&-0.035377&0.618949&3.304019&0.533288    \\
   &&0.3&&-0.022150&0.531732&3.132057&0.546530 \\ 
  &&0.4&&  -0.008924&0.449845&2.962127&0.559904    \\

  1& 1&0.1&0.5 &-0.048603 &0.711494&3.679222&0.520377  \\
  &&&1&-0.048603 &0.711494&3.478013&0.520178    \\
  &&&1.5&-0.048603 &0.711494&3.276718&0.519980\\
  &&&2&-0.048603 &0.711494& 3.075339&0.519782     \\\hline
  
\end{tabular}
\end{center}
\end{table*} 
\begin{table*}[!]
\begin{center}
\caption{ Values of $|\Theta'(0)|$ and $|\Phi'(0)|$ against $\lambda=0.1$, $Fr=1$,  $K=1$, $Pr=6.2$ considering $Q_1^*=Q_2^*=-0.5$ and $Q_1^*=Q_2^*=0.5$ when $V_r<1$.}\label{tb3}
\vspace* {0.0cm}
\begin{tabular}
{ p{0.2in} p{0.2in} p{0.6in}  p{1.2in}   p{1.2in}   p{1.2in}  p{0.7in}  p{0.7in}  p{0.7in}  p{0.7in} }\hline
&&&\hspace{1.5cm}Case-1&&\hspace{1.5cm}Case-2 \\
&&&\hspace{1.2cm}{$Q_1^*=Q_2^*<0$}&&\hspace{1.2cm}{$Q_1^*=Q_2^*>0$}\\
 \hline
$Nt$&$Nb$& $Sc$ &  $|\Theta'(0)|$ &$|\Phi'(0)|$ & $|\Theta'(0)|$ &$|\Phi'(0)|$   \\  \hline

 0.1&0.2&1          &1.472495 &0.839215       & 1.058279&0.839016 \\
  0.4&&             &1.294153 &0.693505            & 0.878091&0.692713   \\
  0.7&&             &1.118039 &0.563965           & 0.700131&0.562578\\
  0.9&&               &1.001867&0.486587               &0.582729&0.484804 \\

 0.2&0.05&1              &1.498066 &0.470481                    &1.083921&0.468896 \\
  &0.06&                 &1.492373 &0.541229                      &1.078182&0.539909   \\
  &0.08&                 &1.480991 &0.629665                      &1.066708&0.628674\\
  &0.1&                  &1.469613&0.682726                       &1.055239&0.681933  \\
 
0.2&0.2&2             &1.412281&0.795764                    &0.997450&0.795368\\
&&2.4                 &1.412073&0.795493                    &0.997242&0.795097\\
  &&2.7               &1.411918 &0.794151                   &0.997086&0.793755\\
  &&3                 &1.411762 &0.791833                   &0.996931&0.791437\\\hline
  \end{tabular}
\end{center}
\end{table*} 
\clearpage
\begin{table*}[!]
\begin{center}
\caption{ Values of $|\Theta'(0)|$ and $|\Phi'(0)|$ against $\lambda=0.1$, $Fr=1$,  $K=1$, $Pr=6.2$ considering $Q_1^*=Q_2^*=-0.5$ and $Q_1^*=Q_2^*=0.5$ when $V_r>1$.}\label{tb4}
\vspace* {0.0cm}
\begin{tabular}
{ p{0.2in} p{0.2in} p{0.6in}  p{1.2in}   p{1.2in}   p{1.2in}  p{0.7in}  p{0.7in}  p{0.7in}  p{0.7in} }\hline
&&&\hspace{1.5cm}Case-1&&\hspace{1.5cm}Case-2 \\
&&&\hspace{1.2cm}{$Q_1^*=Q_2^*<0$}&&\hspace{1.2cm}{$Q_1^*=Q_2^*>0$}\\
 \hline
$Nt$&$Nb$& $Sc$ &  $|\Theta'(0)|$ &$|\Phi'(0)|$ & $|\Theta'(0)|$ &$|\Phi'(0)|$   \\  \hline

 0.1&0.2&1             &4.230189&0.600985                   &3.828641&0.600787 \\
  0.4&&                &3.784521 &0.365738                  &3.381127&0.364945   \\
  0.7&&                &3.341081&0.146660                   &2.935841&0.145273\\
  0.9&&                  &3.046692&0.009590                   &2.640221&0.007807   \\

0.2&0.05&1             &4.263726 &0.023330            &3.862248&0.021746 \\
  &0.06&               &4.251561 &0.133873              &3.850037&0.132552 \\
  &0.08&               &4.227235 &0.272052             &3.825620&0.271061\\
  &0.1&                &4.202915 &0.354959              &3.801208&0.354166  \\

0.2&0.2 &2        &4.071485&0.307618                  &3.669321&0.307221 \\
        &&2.4          &4.067524&0.214874                  &3.665360&0.214478 \\
        &&2.7        &4.064554&0.142511                 &3.662390&0.142115    \\
        &&3          &4.061584&0.067743                  &3.659420&0.067347\\\hline
\end{tabular}
\end{center}
\end{table*}

\end{document}